\documentclass[11pt,a4paper]{article}
\usepackage{amsmath,amssymb,graphicx,latexsym,amsthm,xspace}
\usepackage{algorithm}
\usepackage[noend]{algorithmic}

\newtheorem{theorem}{Theorem}[section]
\newtheorem{lemma}[theorem]{Lemma}
\newtheorem{corollary}[theorem]{Corollary}
\newtheorem{proposition}[theorem]{Proposition}

\newtheorem{claim}[theorem]{Claim}

\newtheorem{definition}[theorem]{Definition}

\title{Constant-Time Algorithms for Sparsity Matroids}
\author{Hiro Ito\thanks{
    School of Informatics, Kyoto University.
    \texttt{itohiro@kuis.kyoto-u.ac.jp}
  }
  \and
  Shin-ichi Tanigawa\thanks{
    Research Institute for Mathematical Sciences, Kyoto University.
    Supported by Grant-in-Aid for JSPS Research Fellowships for Young Scientists.
    \texttt{tanigawa@kurims.kyoto-u.ac.jp}
  }
  \and
  Yuichi Yoshida\thanks{
    School of Informatics, Kyoto University, and 
  Preferred Infrastructure, Inc. 
  Supported by MSRA Fellowship 2010.
  \texttt{yyoshida@kuis.kyoto-u.ac.jp}}
}
  
\newcommand{\calA}{\mathcal{A}}
\newcommand{\calM}{\mathcal{M}}
\newcommand{\calO}{\mathcal{O}}

\newcommand{\calF}{\mathcal{F}}
\newcommand{\calP}{\mathcal{P}}

\begin{document}
\maketitle

\begin{abstract}
  A graph $G=(V,E)$ is called $(k,\ell)$-full if $G$ contains a subgraph $H=(V,F)$ of $k|V|-\ell$ edges such that,
  for any non-empty $F' \subseteq F$, $|F'| \leq k|V(F')| - \ell$ holds.
  Here, $V(F')$ denotes the set of vertices incident to $F'$.
  It is known that the family of edge sets of $(k,\ell)$-full graphs forms a family of matroid, known as the sparsity matroid of $G$.
  In this paper, we give a constant-time approximation algorithm for the rank of the sparsity matroid of a degree-bounded undirected graph. This leads to a constant-time tester for $(k,\ell)$-fullness in the bounded-degree model, 
  (i.e.,  we can decide  with high probability whether an input graph satisfies a property $P$ or far from $P$). 
Depending on the values of $k$ and $\ell$, it can test various properties of a graph such as connectivity, rigidity, and how many spanning trees can be packed.

Based on this result,  we also propose a constant-time tester for $(k,\ell)$-edge-connected-orientability in the bounded-degree model,
where an undirected graph $G$ is called $(k,\ell)$-edge-connected-orientable if there exists an orientation $\vec{G}$ of $G$ with a vertex $r \in V$ such that $\vec{G}$ contains $k$ arc-disjoint dipaths from $r$ to each vertex $v \in V$ and $\ell$ arc-disjoint dipaths from each vertex $v \in V$ to $r$.

  A tester is called a one-sided error tester for $P$ if it always accepts a graph satisfying $P$.
  We show, for $k \geq 2$ and (proper) $\ell \geq 0$,
  any one-sided error tester for $(k,\ell)$-fullness and $(k,\ell)$-edge-connected-orientability requires $\Omega(n)$ queries.
\end{abstract}

\section{Introduction}\label{sec:intro}
Property testing is a relaxation of decision.
In property testing, given an instance $I$,
we are to distinguish the case in which $I$ satisfies a predetermined property $P$ from the case in which $I$ is ``far'' from satisfying $P$.
The farness depends on each model.
The main objective of property testing is to develop efficient algorithms running even in constant time,
which is independent of sizes of instances.

In this paper, we study about testing algorithms for two strongly related properties of undirected graphs, 
$(k, \ell)$-sparsity and $(k,\ell)$-edge-connected-orientability.
A graph $G = (V, E)$ is called \textit{$(k, \ell)$-sparse} if $|F| \leq k|V(F)| - \ell$ for any $F \subseteq E, |F| \geq 1$,
where $V(F)$ denotes the set of vertices incident to edges in $F$.
We note that $(k,\ell)$-sparsity becomes meaningful only when $2k - \ell \geq 1$.
If otherwise, any non-empty graph cannot be $(k,\ell)$-sparse since just an edge violates the condition.
A graph $G$ is called \textit{$(k, \ell)$-tight} if $G$ is $(k, \ell)$-sparse and $|E| = kn - \ell$, 
 where $n$ is the number of vertices in $G$.
A graph $G$ is called \textit{$(k, \ell)$-full} if $G$ contains a $(k,\ell)$-tight subgraph with $n$ vertices.
Checking whether a given graph is $(k,\ell)$-full or not is one of main topics in this paper.

Another topic studied in this paper is an orientability of undirected graphs. 
A (di)graph is called \textit{$k$-edge-connected} (resp., \textit{$k$-vertex-connected}) if deletion of any $k-1$ edges (resp., vertices) leaves the graph connected.
By Menger's theorem, this is equivalent to asking $k$ edge-disjoint (resp., $k$ openly-disjoint) paths between any pair of vertices.
A digraph $D=(V,A)$ is called {\em $(k,\ell)$-edge-connected} with a root $r\in V$ if, for each $v\in V\setminus \{r\}$, $D$ has $k$ arc-disjoint dipaths from $r$ to $v$ and 
$\ell$ arc-disjoint dipaths from $v$ to $r$.
An undirected graph $G=(V,E)$ is called {\em $(k,\ell)$-edge-connected-orientable} ($(k,\ell)$-ec-orientable, in short)
if one can assign an orientation to each edge so that the resulting digraph is $(k,\ell)$-edge-connected with some root $r\in V$.
Note that the choice of $r$ is actually not important, and we may specify any vertex as $r$. 

Nash-Williams' graph-orientation theorem~\cite{nash1960} implies that a graph $G$ admits an orientation such that the resulting digraph is $k$-edge-connected if and only if 
$G$ is $2k$-edge-connected.
This implies that $(k,k)$-ec-orientability of a graph is equivalent to $2k$-edge-connectivity.
Another famous result of Nash-Williams~\cite{nash1964} for the forest-partition problem shows that an undirected graph $G$ contains $k$ edge-disjoint spanning trees if and only if 
$G$ is $(k,k)$-full. 
This theorem, combined with Edmonds' arc-disjoint branching theorem~\cite{edmonds1972}, implies that $G$ is $(k,0)$-ec-orientable if and only if $G$ is $(k,k)$-full.
In this sense, $(k,\ell)$-ec-orientability can be considered as an unified concept of the sparsity and the conventional edge-connectivity.

In this paper, we give constant-time testers for $(k,\ell)$-fullness and $(k,\ell)$-ec-orientability in the bounded-degree model.
In \textit{the bounded-degree model} with a degree bound $d$~\cite{GR02}, 
we only consider graphs with maximum degree at most $d$.
A graph $G = (V, E)$ is represented by an oracle $\calO_G$.
Given a vertex $v$ and an index $i \in \{1,\ldots,d\}$, 
$\calO_G$ returns the $i$-th edges incident to $v$.
If there is no such vertex, $\calO_G$ returns a special character $\bot$.
It can be seen that $\calO_G$ represents the incidence list of $G$,
and we can see one entry of the incidence list by one query to $\calO_G$.
A graph is called \textit{$\epsilon$-far} from a property $P$ if we must modify at least $\frac{\epsilon dn}{2}$ edges.
In other words, we must modify at least $\epsilon$-fraction of the incidence list to make $G$ satisfy $P$.
The \textit{query complexity} of an algorithm is the number of accesses to $\calO_G$.
For a property $P$,
an algorithm is called a \textit{tester} for a property $P$ if it accepts graphs satisfying $P$ with probability at least $\frac{2}{3}$ and rejects graphs $\epsilon$-far from $P$ with probability at least $\frac{2}{3}$.

Our main results are summarized as follows.
\begin{theorem}\label{thr:test-k-l-fullness}
Let $k \geq 1, \ell \geq 0$ be integers with $2k-\ell \geq 1$.
In the bounded-degree model with a degree bound $d$,  
there is a testing algorithm for the $(k,\ell)$-fullness of a graph with query complexity $(k+d)^{O(1/\epsilon'^2)}(\frac{1}{\epsilon'})^{O(1/\epsilon')}$, where $\epsilon'=\frac{\epsilon}{k+d\ell}$.
\end{theorem}
\begin{theorem}\label{thr:test-k-l-orientability}
Let $k \geq 1, \ell \geq 0$ be integers.
In the bounded-degree model with a degree bound $d$,  
there is a testing algorithm for the $(k,\ell)$-ec-orientability of a graph with query complexity  $(k+d)^{O(1/\epsilon'^2)}(\frac{1}{\epsilon'})^{O(1/\epsilon')}$, where $\epsilon'=\max(\frac{\epsilon}{dk},\frac{d\epsilon }{\ell})$.
\end{theorem}
The second result resolves an open problem raised by Orenstein~\cite{Ore10}, 
which asks the existence of a constant-time tester for $(k,\ell)$-ec-orientability.
As mentioned below, the first result has numerous applications to both theoretical and practical problems.   

An algorithm is called a \textit{$(1,\beta)$-approximation algorithm} for a value $x^*$ if, with probability $\frac{2}{3}$,
it outputs $x$ such that $x^* - \beta \leq x \leq x^*$.
For a graph $G=(V,E)$, it is known that the family of edge sets of $(k,\ell)$-sparse subgraphs forms a family of independent sets of a matroid on $E$.
This matroid is called the {\em $(k,\ell)$-sparsity matroid}  of $G$, denoted by ${\cal M}_{k,\ell}(G)$, and the rank function by $\rho_{k,\ell}:2^E\rightarrow \mathbb{Z}$.
Although the detailed property will be discussed in the next section, 
we should note that $G$ is $(k,\ell)$-full if and only if $\rho_{k,\ell}(E)=kn-\ell$. 
To test $(k,\ell)$-fullness, we actually develop a constant-time $(1,\epsilon n)$-approximation algorithm for $\rho_{k,\ell}(E)$.

For a property $P$,
a tester is called a \textit{one-sided error tester} for $P$ if it always accepts graphs satisfying $P$.
A general tester is sometimes called a \textit{two-sided error tester} for comparison.
Our testers for $(k,\ell)$-fullness and $(k,\ell)$-ec-orientability are two-sided error testers.
On the contrary, we give the following lower bounds for one-sided error testers.
\begin{theorem}\label{thr:lower-k-l-fullness}
  Let $k \geq 2, \ell \geq 0$ be integers with $2k - \ell \geq 1$.
  In the bounded-degree model,
  any one-sided error tester for $(k,\ell)$-fullness requires $\Omega(n)$ queries where $n$ is the number of vertices in an input graph.
\end{theorem}
\begin{theorem}\label{thr:lower-k-l-ec-orientability}
  Let $k \geq 2, \ell \geq 0$ be integers with $k > \ell$.
  In the bounded-degree model,
  any one-sided error tester for $(k,\ell)$-ec-orientability requires $\Omega(n)$ queries where $n$ is the number of vertices in an input graph.
\end{theorem}
It is not hard to show that there are one-sided error testers for $(1,\ell)$-fullness and $(1,\ell)$-ec-orientability.
Also, we have a one-sided error tester for $(k,\ell)$-ec-orientability when $\ell \geq k$.

We briefly mention why we use the bounded-degree model.
Another famous model for graphs is the \textit{adjacency matrix model},
in which a graph is represented by an oracle $\calO_G$ such that, 
given two vertices $u$ and $v$, 
$\calO_G$ answers whether there is an edge between $u$ and $v$.
A graph $G$ is called $\epsilon$-far from $P$ in this model if we must modify $\frac{\epsilon n^2}{2}$ edges to make $G$ satisfy $P$.
We show that testing $(k,\ell)$-fullness is trivial in this model. 
Note that we can make any graph $(k,\ell)$-full by adding $kn - \ell$ edges.
Thus, any graph is at most $O(\frac{1}{n})$-far.
Thus, for any $\epsilon > 0$, when $n = \Omega(\frac{1}{\epsilon})$,
we can safely accept graphs without any computation.
When $n = O(\frac{1}{\epsilon})$, we can test $(k,\ell)$-fullness using a standard polynomial-time algorithm.
We have the same issue also for $(k,\ell)$-ec-orientability.

\paragraph{Related works}
In the bounded-degree model, many testers are known for several fundamental graph properties (see e.g.,\cite{goldreich2010}).
The most relevant works are testers for connectivity.
For undirected graphs, Goldreich and Ron~\cite{GR02} gave constant-time testers for $k$-edge-connectivity ($k \geq 1$),
$2$-vertex-connectivity, and $3$-vertex-connectivity.
Yoshida and Ito~\cite{YI08} extended the result by showing constant-time testers for $k$-vertex-connectivity ($k \geq 1$).
For digraphs, 
constant-time testers for $k$-edge-connectivity ($k \geq 1$) are given in~\cite{YI10}.
Recently, Orenstein~\cite{Ore10} simplified those results,
and he also gave constant-time testers for $k$-vertex-connectivity of digraphs ($k \geq 1$).
We stress that the idea behind all the algorithms above is to detect a small evidence that a graph does not satisfy the property we are concerned with.
However, as we discuss later, for $(k,\ell)$-sparsity and $(k,\ell)$-ec-orientability,
there may not be any such small evidence.
This fact makes our testers more involved.

Regarding exact and deterministic algorithms for checking the $(k,\ell)$-fullness of a graph $G$ with $n$ vertices and $m$ edges,
Imai~\cite{Ima83} proposed an algorithm for computing the rank of ${\cal M}_{k,k}(G)$ in $O(n^2)$ time and that of ${\cal M}_{k,\ell}(G)$ in $O(nm)$ time for general $\ell$.
Improved algorithms were proposed by Gabow and Westermann~\cite{gabow:1992}, which run in $O(n\sqrt{m+n\log n})$ time for $k=\ell$ and in $O(n^2)$ time for $k=2$ and $\ell=3$.
Also, they proposed an $O(n\sqrt{n\log n})$-time algorithm for checking the $(2,3)$-tightness (but not fullness). 
An efficient and practical algorithm for computing the rank of ${\cal M}_{k,\ell}(G)$ for general $k$ and $\ell$ is the so-called pebble algorithm by Lee and Streinu~\cite{lee:streinu:2005}, which runs in $O(n^2)$ time.

As the $(k,\ell)$-sparsity has a wide range of applications in rigidity theory and scene analysis (see e.g.~\cite{whiteley:hand,Whiteley:1997}), 
it is recognized as an important open problem to improve the $O(n^2)$ upper-bound for computing the rank of the $(k,\ell)$-sparsity matroid (see e.g.,~\cite[Open Problem 4.1]{demaine2008}). 
To the best of our knowledge, our result is the first sub-quadratic algorithm for approximating the rank of the $(k,\ell)$-sparsity matroid.

\paragraph{Applications}
It is elementary to see that a graph is a forest if and only if it is $(1,1)$-sparse, 
and  the concept of $(1,1)$-fullness coincides with the connectivity of graphs.
As a variant of the commonly studied trees or forests, a graph is called a {\em pseudoforest} if each connected
component contains at most one cycle~\cite{gabow:1992}. 
It is known that a graph is a pseudoforest if and only if it is $(1,0)$-sparse~\cite{gabow:1988}.
As we mentioned above, Nash-Williams~\cite{nash1964} proved that a graph contains $k$ edge-disjoint spanning trees if and only if it is $(k,k)$-full.
Motivated by an application to rigidity theory, Whiteley~\cite{Whiteley:1997} or Haas~\cite{haas:2002} proved a generalization of Nash-Williams' theorem to $(k,\ell)$-sparse graphs by mixing trees and pseudoforests.
Our result leads to constant-time testers for these properties.

Another important application of $(k,\ell)$-sparse graphs is the rigidity of graphs.
A classical theorem by Laman~\cite{laman:1970} implies that a $(2,3)$-full graph has a special property of being a generically rigid bar-joint framework on the plane, 
by regarding each vertex as a joint and each edge as a bar.
More precisely, the deficiency between $2n-3$ and the rank of the $(2,3)$-sparsity matroid is equal to the degree of freedoms of the graph in the plane. 
It is further proved by Whiteley~\cite{whiteley:88} that the $(2,2)$-sparsity matroid characterizes the generic rigidity of graphs embedded on a torus or a cylinder 
while the $(2,1)$-sparsity does the generic rigidity of graphs on the surface of a cone.
For a general $d$-dimensional case, the $({d+1 \choose 2}, {d+1\choose 2})$-sparsity matroid characterizes the generic rigidity of special types of structural models, 
called body-bar frameworks~\cite{tay:84} and body-hinge frameworks~\cite{whiteley:88}. 

Although a combinatorial characterization of 3-dimensional generic rigidity of graphs  has not been found yet (see e.g.~\cite{whiteley:hand,Whiteley:1997}), 
a characterization of an important special class, called {\em molecular graphs}, has been proved recently.
In terms of graph theory, a molecular graph means the square $G^2$ of a graph $G$
as the rigidity of a molecule can be modeled by the rigidity of the square of a graph by identifying each atom as a vertex and each covalent bond as an edge (see e.g.,~\cite{thorpe:2005,Whiteley:2005}).
Tay and Whiteley~\cite{tay:whiteley:84}, or more formally, Jackson and Jord{\'a}n~\cite{jackson:08},  conjectured that $G^2$ is generically rigid in 3-dimensional space if and only if $5G$ is $(6,6)$-full.
Here $5G$ denotes the graph obtained from $G$ by duplicating each edge by five parallel copies.
Recently, Katoh and Tanigawa~\cite{molecular} solved this conjecture affirmatively.
In fact, based on this theory, the pebble game algorithm for checking $(6,6)$-fullness (runs in $O(n^2)$ time) is implemented in several softwares
(e.g.,~\cite{amato,flexweb,kinari}) to compute the degree of freedoms of proteins. 
In this sense our super-efficient approximation algorithm for computing the degree of freedoms of molecules could bring a totally new approach in the protein flexibility analysis and the similarity search in the protein data base.

\paragraph{Organization and proof overview}
In Section~\ref{sec:pre}, we review properties of ${\cal M}_{k,\ell}(G)$.
Then, in Sections~\ref{sec:matching} and~\ref{sec:sparsity}, we first describe how to test $(k,\ell)$-fullness.
To test whether $G$ is $(k,\ell)$-full,
we develop a $(1,\epsilon n)$-approximation algorithm for $\rho_{k,\ell}(E)$ running in constant time (Theorem~\ref{thr:approximation-to-mkl}).

A natural way to estimate the rank of $\calM_{k,\ell}(G)$ is locally simulating the greedy algorithm,
i.e., we add edges one by one,
and if a newly added edge forms a circuit, we discard it.
The main obstacle to simulate this algorithm is that, in general, we cannot detect any circuit in constant time.
For example, a circuit in $\calM_{1,1}(G)$ corresponds to a cycle in $G$.
However, there is a $d$-regular graph in which any cycle is of length $\Omega(\log_d n)$.
Thus, we need to estimate the rank without seeing any circuit.
We mention that, for $\calM_{1,1}(G)$, 
it is known that $\rho_{1,1}(E) = n-c$ holds where $c$ is the number of connected components.
Using this fact,~\cite{CRT01} gave an algorithm to estimate $\rho_{1,1}(E)$.
However, for general $k$ and $\ell$, there is no such formula.

Our strategy to overcome this issue is as follows:
First, we remove constant-size circuits w.r.t.\ $\calM_{k,\ell}(G)$,
and let $G' = (V, E')$ be the resulting graph.
We can show that $\rho_{k,\ell}(E) = \rho_{k,\ell}(E')$.
A crucial fact is that $\rho_{k,\ell}(E')$ is close to $\rho_{k,0}(E')$.
Thus, it amount to estimate $\rho_{k,0}(E')$ efficiently.
It is known that $\rho_{k,0}(E')$ equals the size of the maximum matching of an auxiliary graph,
and we can compute the maximum matching with a constant-time approximation algorithm for the maximum matching~\cite{YYI09}.

In Section~\ref{sec:orientability}, we provide a constant-time tester for $(k,\ell)$-ec-orientability.
Our algorithm is based on a characterization of the number of edges we need to add to make a graph $(k,\ell)$-ec-orientable by Frank and Kir{\'a}ly~\cite{FK03}.
Although this characterization is not so simple as the case of the edge-connectivity augmentation problem, 
we are able to show that, if $G$ is $\epsilon$-far, either there are many small evidences or $G$ is globally sparse which can be measured by $(k,k)$-fullness (Theorem~\ref{thm:key2}). 
As mentioned in introduction, the $(k,\ell)$-ec-orientability has strong relations to the sparsity as well as to the conventional edge-connectivity.
Indeed, our algorithm can be considered as a combination of the idea of Yoshida and Ito for testing connectivity and the algorithm for testing $(k,k)$-sparsity given in Section~\ref{sec:sparsity}.

In Section~\ref{sec:lower-bound}, we prove linear lower bounds of one-sided error testers.
In~\cite{Ore10}, Orenstein proved linear lower bounds of one-sided error tester for $(k,0)$-ec-orientability (or equivalently, $(k,k)$-fullness).
Orenstein's proof made use of Tutte-and-Nash-Williams' tree packing theorem (see Theorem~\ref{thm:tutte}), which is a special property of $(k,k)$-fullness.
We can however show that Orenstein's approach can be applied to the general case of $\ell$ by the use of graph operations that preserve $(k,\ell)$-fullness.

\section{Preliminaries}\label{sec:pre}
For an integer $n$, we denote by $[n]$ the set $\{1,\ldots,n\}$.
Let $G = (V,E)$ be a graph.
For a vertex set $S \subseteq V$, $G[S]$ denote the subgraph of $G$ induced by $S$.
For an edge set $F \subseteq E$, 
we define $V_G(F)$ as the set of vertices incident to $F$.
For a vertex set $S,T \subseteq V$, 
we define $E_G(S,T) = \{uv \in E\mid u \in S, v \in T\}$ and $d_G(S,T)=|E_G(S,T)|$.
If $T=V\setminus S$, we abbreviate them as $E_G(S)$ and $d_G(S)$, respectively. 
For a vertex $v \in V$, we use $E_G(v)$ and $E_G(v,T)$ instead of $E_G(\{v\})$ and $E_G(\{v\},T)$, respectively.
Also, we define $\Gamma_G(S)$ as the set of vertices in $V \setminus S$ adjacent to some vertex in $S$.
When the context is clear, we omit the subscripts.

Let $f:2^E\rightarrow \mathbb{R}$ be a set function on a finite set $E$.
$f$ is called {\em submodular} if $f(X)+f(Y)\geq f(X\cap Y)+f(X\cup Y)$ holds for any $X,Y\subseteq V$,
and $f$ is called {\em non-decreasing} if $f(X)\leq f(Y)$ for any $X\subseteq Y\subseteq V$.
Edmonds and Rota~\cite{edmonds:1966} observed (and Pym and Perfect~\cite{pym} formally proved) that an integer-valued non-decreasing submodular function $f:2^E\rightarrow \mathbb{Z}$
{\em induces} a matroid on $E$, where $F\subseteq E$ is independent if and only if $|F'|\leq f(F')$ for every non-empty $F'\subseteq F$.

For a graph $G=(V,E)$ and integers $k \geq 1,\ell \geq 0$, we define a function $f_{k,\ell}:2^E\rightarrow \mathbb{Z}$ by $f_{k,\ell}(F)=k|V(F)|-\ell$ for $F\subseteq E$.
It is known (and easy to show anyway) that $f_{k,\ell}$ is non-decreasing and submodular.
Thus, $f_{k,\ell}$ induces a matroid on $E$, that is, the $(k,\ell)$-sparsity matroid ${\cal M}_{k,\ell}(G)$ defined in introduction.
The rank function  and the closure operator are denoted by $\rho_{k,\ell}$ and ${\rm cl}_{k,\ell}$, respectively.
We note that $\rho_{k,\ell}(F)$ equals the size of the largest $(k,\ell)$-sparse edge set contained in $F$.
This implies that $G$ is $(k,\ell)$-tight iff the rank of ${\cal M}_{k,\ell}(G)$ is $kn-\ell$.

A set $F\subseteq E$ is called a \textit{$(k,\ell)$-connected set} if, for any pair $e,e'\in F$, $F$ has a circuit of ${\cal M}_{k,\ell}(G)$ that contains $e$ and $e'$. 
For simplicity of the description, a singleton $\{e\}$ is also considered as a $(k,\ell)$-connected set.
A maximal $(k,\ell)$-connected set w.r.t.\ edge inclusion is called a \textit{$(k,\ell)$-connected component}.
The following property of $(k,\ell)$-connected sets is just a restatement of a general fact on matroid-connectivity for our purpose.
\begin{proposition}
\label{prop:property1} 
Let $G=(V,E)$ be a graph and $k \geq 1, \ell \geq 0$ be integers with $2k-\ell\geq 1$.
Then, ${\cal M}_{k,\ell}(G)$ has the following properties:
\begin{description}
\item[(i)] For two $(k,\ell)$-connected sets $F_1$ and $F_2$ with $F_1\cap F_2\neq \emptyset$, $F_1\cup F_2$ is $(k,\ell)$-connected.
\item[(ii)] We can uniquely partition $E$ into $(k,\ell)$-connected components $\{C_1,\dots, C_t\}$,
and the following relation holds: 
\begin{equation}
\label{eq:connected_component}
\rho_{k,\ell}(E)=\sum_{i=1}^t \rho_{k,\ell}(C_i).
\end{equation}
\end{description}
\end{proposition}
Proofs can be found in e.g.,~\cite[Chapter~4]{Oxl92}.
A $(k,\ell)$-connected set (or component) is called \textit{trivial} if it is singleton, otherwise \textit{non-trivial}. 
We remark that $\{e\}$ is a trivial $(k,\ell)$-connected component if and only if $e$ is a \textit{coloop} in ${\cal M}_{k,\ell}(G)$ (i.e., every base contains $e$) 
since ${\cal M}_{k,\ell}(G)$ has no loop (in the matroid sense) if $2k-\ell\geq 1$.
Hence, if we denote the family of non-trivial $(k,\ell)$-connected components in ${\cal M}_{k,\ell}(G)$ by $\{C_1,\dots, C_s\}$, then 
\eqref{eq:connected_component} implies 
\begin{equation}
\label{eq:connected_component2}
\rho_{k,\ell}(E)=\left|E\setminus \bigcup_{i=1}^s C_i\right|+\sum_{i=1}^s \rho_{k,\ell}(C_i).
\end{equation}
  
We also need the following known properties of ${\cal M}_{k,\ell}(G)$.
(Since they are so fundamental, we present proofs for completeness.) 
\begin{lemma}
\label{lmm:property2}
Let $G=(V,E)$ be a graph and $k \geq 1, \ell \geq 0$ be integers with $2k-\ell\geq 1$.
Then, ${\cal M}_{k,\ell}(G)$ has the following properties:
\begin{description}
\item[(i)] For any circuit $C$ of ${\cal M}_{k,\ell}(G)$, $\rho_{k,\ell}(C)=f_{k,\ell}(C)$.
\item[(ii)] For any non-trivial $(k,\ell)$-connected set $F\subseteq E$, $\rho_{k,\ell}(F)=f_{k,\ell}(F)$. 
Namely, $F$ is $(k,\ell)$-full.
\end{description}
\end{lemma}
\begin{proof}
A proof for (i): 
Since $C$ is a minimal dependent set, $|C|>f_{k,\ell}(C)$ and $|C|-1=|C-e|\leq f_{k,\ell}(C-e)\leq f_{k,\ell}(C)$ for any $e\in C$.
This implies $|C|=f_{k,\ell}(C)+1$. Thus, $\rho_{k,\ell}(C)=|C|-1=f_{k,\ell}(C)$.

A proof for (ii): Suppose $\rho_{k,\ell}(F)<f_{k,\ell}(F)$. 
Then, there is an edge $uv\notin F$ with $u,v\in V(F)$ such that $\rho_{k,\ell}(F+uv)=\rho_{k,\ell}(F)+1$.
Let us take two distinct edges $e$ and $e'$ of $F$ incident to $u$ and $v$, respectively.
(It is easy to see that such two edges exist since $F$ is $(k,\ell)$-connected.)
Since $F$ is $f$-connected, there is a circuit $C\subseteq F$ that contains $e$ and $e'$.
Then, by (i) and by $f_{k,\ell}(C+uv)=f_{k,\ell}(C)$,  
we obtain $\rho_{k,\ell}(C+uv)\leq f_{k,\ell}(C+uv)=f_{k,\ell}(C)=\rho_{k,\ell}(C)$, implying $\rho_{k,\ell}(C+uv)=\rho_{k,\ell}(C)$.
In other words, $uv$ is contained in the closure of $C$. This contradicts $\rho_{k,\ell}(F+uv)=\rho_{k,\ell}(F)+1$.
\end{proof}

We also need the following relation between ${\cal M}_{k,\ell}(G)$ and ${\cal M}_{k,\ell'}(G)$ with distinct $\ell$ and $\ell'$.
\begin{lemma}
\label{lmm:property3}
Let $G=(V,E)$ be a graph and $k \geq 1, \ell \geq 0$ be integers with $2k-\ell\geq 1$.
Then, any $(k,\ell)$-sparse set $F\subseteq E$ is $(k,\ell')$-sparse for any $\ell'\leq \ell$.
\end{lemma}
\begin{proof}
For any nonempty $F'\subseteq F$, we have $|F'|\leq k|V(F')|-\ell\leq k|V(F)|-\ell'$.
\end{proof}

Finally, we give the formal definition of the bounded-degree model.
\begin{definition}[Bounded-degree model] 
  In the \textit{bounded-degree model} with a degree bound $d$,
  we consider graphs with maximum degree at most $d$.
  A graph $G=(V,E)$ of $n$ vertices is represented by an oracle $\calO_G$ satisfying the followings:
  \begin{itemize}
    \setlength{\itemsep}{0pt}
  \item For each vertex $v \in V$, there exists an injection $\pi_v: E_G(v) \to [d]$
    such that $\pi_v$ is an injection.
  \item The oracle $\calO_G$, on two numbers $u \in V, i\in \mathbb{N}$,
    returns $v$ such that $(u,v)\in E$ and $\pi_u((u,v)) = i$.
    If no such vertex $v$ exists, it returns a special character $\bot$.
    An edge $e = uv$ is called the $i$-th edge of $u$ if $\pi_u(e) = i$.
  \end{itemize}
  Algorithms are given $V$, $n$, $d$, and the access to $\calO_G$ beforehand.
  For an error parameter $\epsilon > 0$, 
  a graph is called $\epsilon$-far from a property $P$,
  if we must add or remove at least $\frac{\epsilon dn}{2}$ edges to make $G$ satisfy $P$. 
\end{definition}

\section{Approximating the rank of ${\cal M}_{k,0}(G)$}\label{sec:matching}
\label{sec:k0}
Let $k \geq 1$ be an integer.
In this section, 
we present a constant-time approximation algorithm for the rank $\rho_{k,0}(E)$ of $\calM_{k,0}(G)$ for a graph $G=(V,E)$.
A crucial fact is that computing $\rho_{k,0}(E)$ can be reduced to computing the size of a maximum matching in an auxiliary bipartite graph $G_k$ obtained from $G$.
The vertex set of $G_k$ is $E \cup (V \times [k])$ where $E$ and $V \times [k]$ form a partition, and
$G_k$ has an edge between $e \in E$ and $(v, i) \in V \times [k]$ iff  $e$ is incident to $v$ in the original graph $G$ (see Figure~\ref{fig:auxiliary}). 
From the celebrated Hall's marriage theorem, 
the following result easily follows (see e.g.,~\cite{Ima83} for more details):

\begin{figure}[t]
  \centering
  \begin{minipage}{0.4\textwidth}
    \centering
    \includegraphics[scale=0.8]{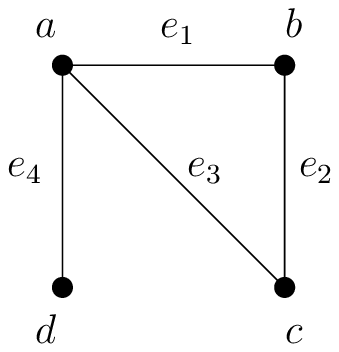}
    \par (a)
  \end{minipage}
  \begin{minipage}{0.4\textwidth}
    \centering
    \includegraphics[scale=0.8]{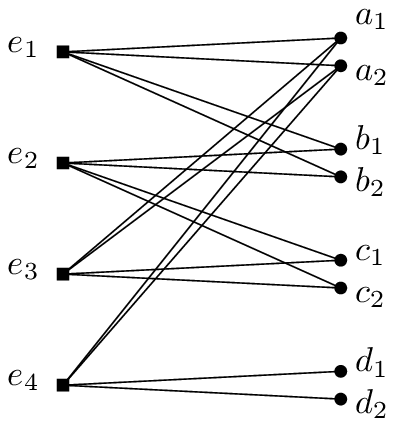}
    \par (b)
  \end{minipage}
  \caption{(a) $G$ and (b) $G_2$.}
  \label{fig:auxiliary}
\end{figure}

\begin{proposition}
\label{prop:pseudo_forest}
Let $G=(V,E)$ be a graph and $k$ be an integer.
Then, $G_k$ contains a matching covering $F \subseteq E$ if and only if $F$ is $(k,0)$-sparse.
\end{proposition}
Proposition~\ref{prop:pseudo_forest} implies that the rank of ${\cal M}_{k,0}(G)$ is equal to the size of a maximum matching in $G_k$.
We use the following algorithm.
\begin{lemma}[\cite{YYI09}]\label{lmm:approximation-to-matching}
  In the bounded-degree model with a degree bound $d$, 
  there exists a $(1,\epsilon n)$-approximation algorithm for the size of the maximum matching of a graph with query complexity $d^{O(1/\epsilon^2)}(\frac{1}{\epsilon})^{O(1/\epsilon)}$.
\end{lemma}

To run the algorithm given in Lemma~\ref{lmm:approximation-to-matching} on $G_k$,
we want to make an oracle access $\calO_{G_k}$ to $G_k$ using the oracle access $\calO_G$ to $G$.
However, since we do not have a method to access $E$ directly,
the vertex set $E \cup (V \times [k])$ is inconvenient to design $\calO_{G_k}$.

To deal with this issue, 
we use a slightly different auxiliary graph,
which is essentially equivalent to the previous auxiliary graph.
First, we introduce \textit{arbitrary} ordering among vertices.
We call $(v, i) \in V \times [d]$ \textit{valid} if the $i$-th edge incident to $v$ exists and the vertex $v$ is the larger one in the endpoints of the edge,
and \textit{invalid} if otherwise.
Then, we define a graph $G_k = (U_k \cup V_k, E_k)$ where
\begin{eqnarray*}
  U_k &=& \{(0, v, i) \mid v \in V, i \in [d]\},\\
  V_k &=& \{(1, v, i) \mid v \in V, i \in [k]\}, \\
  E_k &=& \{ ((0, u, i), (1, v, j)) \in U_k \times V_k \mid \text{if $(u, i)$ is valid and the corresponding edge is incident to $v$ in $G$} \}.
\end{eqnarray*}
For a vertex $(b, v, i)$ in $G_k$, 
the first bit $b$ is used to distinguish whether the vertex is in $U_k$ or $V_k$.
We can see $G_k$ constructed here is isomorphic to the graph obtained from the previous auxiliary graph by adding singleton vertices.

\begin{lemma}\label{lmm:approximation-to-mk0}
  Let $k \geq 1$ be an integer.
  In the bounded-degree model with a degree bound $d$,
  for a graph $G$ of $n$ vertices,
  there exists a $(1,\epsilon n)$-approximation algorithm for the rank of $\calM_{k,0}(G)$ with query complexity $(k+d)^{O(1/\epsilon'^2)}(\frac{1}{\epsilon'})^{O(1/\epsilon')}$ where $\epsilon' = \frac{\epsilon}{k+d}$.
\end{lemma}
\begin{proof}
  Let $n$ and $m$ be the number of vertices and edges in $G = (V, E)$, respectively.
  The number of vertices in $G_k$ is $n' := kn + dn = (k+d)n$.
  Also, the maximum degree of $G_k$ is $d' := \max(2k, d) = O(k + d)$.

  Using the oracle access $\calO_G$ to $G$,
  we make an oracle access $\calO_{G_k}$ to $G_k$ on which we will run the algorithm given in Lemma~\ref{lmm:approximation-to-matching}.
  For a query $\calO_{G_k}((b, v, i), j)$, we do the following.

  Suppose that $b = 0$, which means that $(b, v, i)$ is a vertex in $U_k$.
  We can check whether $(v, i)$ is valid by asking $\calO_G$ once.
  If $(v, i)$ is invalid, we return $\bot$.
  Suppose that $(v, i)$ is valid and it corresponds to an edge $e = uv$ where $v > u$.
  If $j \leq k$, we return $(1, u, j)$.
  If $j > k$, we return $(1, v, j - k + 1)$.
  
  Suppose that $b = 1$, which means that $(b, v, i)$ is a vertex in $V_k$.
  If there is no $j$-th edge incident to $v$, we return $\bot$.
  Let $e = uv$ be the $j$-th edge incident to $v$.
  If $v > u$, we return $(0, v, j)$.
  If $v < u$ and $e$ is the $j'$-th edge of $u$, we return $(0, u, j')$.
  Here, we can find $j'$ by asking $\calO_G$ at most $d$ times.

  To summarize, we can simulate the oracle access $\calO_{G_k}$ by asking $\calO_G$ at most $d+1$ times.
  To approximate the rank of $\calM_{k,0}(G)$ with an additive error $\epsilon n$,
  we run the algorithm given in Lemma~\ref{lmm:approximation-to-matching} on $\calO_{G_k}$
  after replacing $\epsilon$ by $\epsilon' = \frac{\epsilon}{k+d}$.
  The query complexity becomes $d\cdot d'^{O(1/\epsilon'^2)}(\frac{1}{\epsilon'})^{O(1/\epsilon')}$,
  and the lemma holds.
\end{proof}

\section{Approximating the rank of ${\cal M}_{k,\ell}(G)$}\label{sec:sparsity}
In this section, 
we describe a constant-time approximation algorithm for the rank of $\calM_{k,\ell}(G)$ for a graph $G = (V, E)$.
For a given error value $\epsilon$, 
let $t$ be a constant determined later.
We say that a subset $S\subseteq E$ is {\em large} if $|S|\geq t$; otherwise called {\em small}.


For an edge $e = uv$ and an integer $r > 0$, 
let $G_r(e)$ be the graph induced by the set of vertices whose distance to $u$ or $v$ is at most $r$.
Also, let $E_r(e)$ be the set of edges in $G_r(e)$.
The core of our approximation algorithm is an efficient implementation of an algorithm \textsf{Component}$(e)$ that (approximately) decides whether a given edge $e \in E$ is in a large $(k,\ell)$-connected set or not.
As a subroutine, we first prepare an algorithm called \textsf{SmallCircuits}$(e)$ in Algorithm~1 and then show \textsf{Component}$(e)$ in Algorithm~2.

\begin{algorithm}
  \caption{\textsf{SmallCircuits}$(e)$: returns the union of small circuits containing an edge $e$}
  \label{alg:small-curcuits}
  \begin{algorithmic}[1]
    \STATE $S = \{e\}$.
    \WHILE {there is an unchecked small circuit $C\subseteq E_t(e)$ containing $e$}
    \STATE $S=S\cup C$.
    \IF{$|S|\geq t$}
    \STATE \textbf{return} \textbf{Large} (a special symbol).
    \ENDIF
    \ENDWHILE
    \STATE \textbf{return} $S$.
  \end{algorithmic}
\end{algorithm}

\begin{algorithm}
  \caption{\textsf{Component}$(e)$: decides whether $e$ is contained in a large $(k,\ell)$-connected set}
  \label{alg:component}
  \begin{algorithmic}[1]
    \STATE $S = \{e\}$.
    \WHILE {there is an unchecked element $f$ in $S$}
    \IF{$\textsf{SmallCircuits}(f)=\mathbf{Large}$}
    \STATE \textbf{return} \textbf{Large}.
    \ENDIF
    \STATE check $f$. 
    \STATE $S=S\cup \textsf{SmallCircuits}(f)$ \label{line}
    \IF{$|S| \geq t$}
    \STATE \textbf{return} \textbf{Large}.
    \ENDIF
    \ENDWHILE
    \STATE \textbf{return} $S$.
  \end{algorithmic}
\end{algorithm}

The following sequence of lemmas shows structural properties of outputs of \textsf{SmallCircuits}$(e)$ and \textsf{Component}$(e)$.
\begin{lemma}
\label{lmm:circuits}
For any $e\in E$, $\textsf{SmallCircuits}(e)$ and $\textsf{Component}(e)$ are small $(k,\ell)$-connected sets unless they return \textbf{Large}.
\end{lemma}
\begin{proof}
Let $S=\textsf{SmallCircuits}(e)$.
If $S=\{e\}$, then $S$ is a trivial $(k,\ell)$-connected set.
If  $|S|>1$, then $S$ is the union of circuits containing $e$.
By Proposition~\ref{prop:property1}(i), $S$ is $(k,\ell)$-connected.

The latter claim similarly follows from Proposition~\ref{prop:property1}(i)  since $\textsf{Component}(e)$ is the union of $\textsf{SmallCircuit}(e')$ for all $e'\in \textsf{Component}(e)$.
\end{proof}

%

We define a relation $\sim$ on $E$ such that $e\sim f$ for $e,f\in E$ if and only if ${\cal M}_{k,\ell}(G)$ has  a small circuit that contains $e$ and $f$. 

\begin{lemma}
\label{lmm:large}
Suppose that $\textsf{Component}(e)=\mathbf{Large}$.
Then, there is a large $(k,\ell)$-connected set $S$ containing $e$ such that,
for each $f\in S$, 
\begin{itemize}
\item $e\sim f$, or 
\item $e\sim f'\sim f$ for some $f'\in S$.
\end{itemize}
\end{lemma}
\begin{proof}
If $\textsf{SmallCircuits}(e)$ returns \textbf{Large}, then the union of small circuits containing $e$ forms a large $(k,\ell)$-connected set.
This set satisfies the property of the statement.

Thus, assume $\textsf{SmallCircuits}(e)\neq \mathbf{Large}$. 
Since $\textsf{Component}(e)$ returns \textbf{Large}, we encounter either one of the following two situations at the end of Algorithm~2:
a small $(k,\ell)$-connected set $S$ with $e\in S$ contains an edge $f$ such that (i) $\textsf{SmallCircuits}(f)$ returns \textbf{Large} or 
(ii) $\textsf{SmallCircuits}(f)$ is small but $S\cup \textsf{SmallCircuits}(f)$ is large.
In both cases, let $S_{f}$ be the union of all small circuits containing $f$.
Then, $S\cup S_{f}$ is a desired large $(k,\ell)$-connected set. 
\end{proof}

\begin{lemma}
\label{lmm:maximal}
Let $e\in E$.
Suppose that $\textsf{Component}(e)$ does not return \textbf{Large}.
Then, every small $(k,\ell)$-connected set intersecting $\textsf{Component}(e)$ is contained in $\textsf{Component}(e)$.
\end{lemma}
\begin{proof}
Let $S=\textsf{Component}(e)$.
Suppose that ${\cal M}_{k,\ell}(G)$ has a small $(k,\ell)$-connected set $S'$ such that $S\cap S'\neq\emptyset$ and $S'\setminus S\neq \emptyset$.
Take $f\in S\cap S'$ and $f'\in S'\setminus S$.
Since $f\sim f'$, we have $f'\in \textsf{SmallCircuts}(f)$.
By Line~\ref{line} of Algorithm~2, we obtain $f'\in \textsf{Component}(e)=S$, a contradiction. 
\end{proof}

\begin{lemma}
\label{lmm:relation}
For any $e\in E$ with $\textsf{Component}(e)\neq \mathbf{Large}$ and for any $f\in \textsf{Component}(e)$,
$\textsf{Component}(e)=\textsf{Component}(f)$.
\end{lemma}
\begin{proof}
Let $S=\textsf{Component}(e)$.
Suppose that $\textsf{Component}(f)= \mathbf{Large}$.
Then, by Lemma~\ref{lmm:large}, there exists a large $(k,\ell)$-connected set $S_f$ containing $f$ such that, for each $f'\in S_f$, $f\sim f'$ or $f\sim f''\sim f'$ holds for some $f''\in S_f$.
In particular,  $S$ contains every element of $S_f$ since $e\sim f$ and $\textsf{Component}(e)$ never return \textbf{Large} during Algorithm~2.
This contradicts that $S$ is small.

Thus, $\textsf{Component}(f)$ is a small $(k,\ell)$-connected set by Lemma~\ref{lmm:circuits}.
Lemma~\ref{lmm:maximal} now implies $\textsf{Component}(f)\subseteq S$ and $S\subseteq \textsf{Component}(f)$.
\end{proof}

Let $L=\{e\in E\mid \textsf{Component}(e)=\mathbf{Large}\}$, and let $\{S_1,S_2,\dots, S_m\}$ be the set of subsets of $E$ such that $S_i=\textsf{Component}(e)$ for some $e\in E$.
Then, by Lemma~\ref{lmm:relation}, $\{L,S_1,\dots,S_m\}$ forms a partition of $E$.
Our testing algorithm directly follows from the next theorem.
%
\begin{theorem}\label{thm:key}
Let $\{L,S_1,\dots, S_m\}$ be the partition of $E$ defined as above.
For each $i$ with $1\leq i \leq m$,  let $B_i$ be a base of $S_i$ in $\calM_{k,\ell}(G)$, and let $E' = L \cup \bigcup_{i=1}^m B_i$.
Then, $\rho_{k,0}(E')-\frac{\ell dn}{t}\leq \rho_{k,\ell}(E)\leq \rho_{k,0}(E')$.
\end{theorem}
\begin{proof}
Since $B_i$ is a base of $S_i$ in ${\cal M}_{k,\ell}(G)$, we have $S_i\subseteq {\rm cl}_{k,\ell}(B_i)\subseteq {\rm cl}_{k,\ell}(E')$ for each $i$.
This implies $\rho_{k,\ell}(E')=\rho_{k,\ell}(E)$.
Also, by Lemma~\ref{lmm:property3}, we have $\rho_{k,\ell}(E')\leq \rho_{k,0}(E')$.

To see $\rho_{k,0}(E')-\frac{\ell dn}{t}\leq \rho_{k,\ell}(E')$, 
recall that $(k,\ell)$-connected components of ${\cal M}_{k,\ell}(G)|E'$ partitions $E'$ by Proposition~\ref{prop:property1}(ii)
(where ${\cal M}_{k,\ell}(G)|E'$ denotes the restriction of ${\cal M}_{k,\ell}(G)$ to $E'$).
We have the following properties of these connected sets.

\begin{claim}
\label{claim1}
Any $e\in L$ is contained in a large $(k,\ell)$-connected component in  ${\cal M}_{k,\ell}(G)|E'$.
\end{claim}
\begin{proof}
Let us take a large $(k,\ell)$-connected set $S_e$ of ${\cal M}_{k,\ell}(G)$ satisfying the property of Lemma~\ref{lmm:large} for $e$.
Suppose, for a contradiction, that $\textsf{Component}(f)$ returns a small $(k,\ell)$-connected set $S_f$ for some $f\in S_e$.
By a property of $S_e$, for every $f'\in S_e$ we have $f\sim f_1\sim e\sim f_2 \sim f'$ for some  $f_1, f_2\in S_e$.
As $\textsf{Component}(f)$ never return \textbf{Large}, we have $S_e\subseteq S_f$ according to Algorithm 2, contradicting that $S_f$ is small.

Thus,  each element of $S_e$ is included in $L$.
This implies that $S_e$ remains in $E'$.
Namely, $S_e$ exists as a large $(k,\ell)$-connected set even in ${\cal M}_{k,\ell}(G)|E'$, and $e$ is contained in a large $(k,\ell)$-connected component in ${\cal M}_{k,\ell}(G)|E'$.
\end{proof}

\begin{claim}
\label{claim2}
Every non-trivial $(k,\ell)$-connected component in ${\cal M}_{k,\ell}(G)|E'$ is large.
\end{claim}
\begin{proof}
To see this, suppose that there is a non-trivial small $(k,\ell)$-connected component $C$ in ${\cal M}_{k,\ell}(G)|E'$.
By Claim~\ref{claim1}, each element of $L$ belongs to a large $(k,\ell)$-connected component in ${\cal M}_{k,\ell}(G)|E'$.
This implies $C\subseteq \bigcup_{i=1}^m B_i$.
Also, since $B_i$ is independent in ${\cal M}_{k,\ell}(G)$, $C$ must intersect at least two sets among $\{B_1,\dots, B_m\}$.
In particular, $C$ intersects at least two sets among $\{S_1,\dots, S_m\}$.
Since $C$ is a small $(k,\ell)$-connected set in ${\cal M}_{k,\ell}(G)$, this contradicts Lemma~\ref{lmm:maximal}.
\end{proof}

Let $\{C_1,C_2,\dots,C_s\}$ be the family of non-trivial $(k,\ell)$-connected components in ${\cal M}_{k,\ell}(G)|E'$.
Note that $s\leq \frac{dn}{t}$ holds by Claim~\ref{claim2}.
Therefore,
\begin{eqnarray*}
    \rho_{k,\ell}(E')
    &=&
    |E'\setminus \bigcup_{i=1}^s C_i| + \sum_{i=1}^s \rho_{k,\ell}(C_i) \quad \text{(by \eqref{eq:connected_component2})} \\
    &=&
    |E'\setminus \bigcup_{i=1}^s C_i| + \sum_{i=1}^s (k|V(C_i)|-\ell)  \quad \text{(by Lemma~\ref{lmm:property2}(ii))} \\ 
    &\geq &
    |E'\setminus \bigcup_{i=1}^s C_i| + \sum_{i=1}^s k|V(C_i)| - \frac{\ell dn}{t} \quad \text{(by $s\leq \frac{dn}{t}$)}.
  \end{eqnarray*}
On the other hand,
  \begin{eqnarray*}
    \rho_{k,0}(E')
    &\leq&
    |E'\setminus \bigcup_{i=1}^s C_i| + \sum_{i=1}^s \rho_{k,0}(C_i) \quad \text{(by the submodularity of $\rho_{k,0}$)} \\
    &\leq&
    |E'\setminus \bigcup_{i=1}^s C_i| + \sum_{i=1}^s k|V(C_i)|.
  \end{eqnarray*}
This completes the proof.
\end{proof}

\begin{theorem}\label{thr:approximation-to-mkl}
  Let $G=(V,E)$ be a graph with $n$ vertices, and $k \geq 1, \ell \geq 0$ be integers with $2k - \ell \geq 1$.
  In the bounded-degree model with a degree bound $d$,
  there exists a $(1,\epsilon n)$-approximation algorithm for the rank of $\calM_{k,\ell}(G)$ 
with query complexity $(k+d)^{O(1/\epsilon'^2)}(\frac{1}{\epsilon'})^{O(1/\epsilon')}$ where $\epsilon' = \frac{\epsilon}{k+d\ell}$.
\end{theorem}
\begin{proof}
  Let $G' = (V, E')$ where $E'$ is as in Theorem~\ref{thm:key}.
  Set $t=\frac{\ell d}{\epsilon}$.
  Our algorithm computes $\rho_{k,0}(E')$ based on the algorithm given in Lemma~\ref{lmm:approximation-to-mk0} for the error threshold $\epsilon$ and just returns this value.
  By Lemma~\ref{lmm:approximation-to-mk0} and Theorem~\ref{thm:key}, this value approximates $\rho_{k,\ell}(E)$ with additive error $\epsilon n$.
  Therefore, if we can make an oracle access $\calO_{G'}$ to the graph $G'$, we are done.
  
  For a query $\calO_{G'}(v, i)$, we return a value as follows.
  If $\calO_{G}(v, i) = \bot$, we return $\bot$.
  Suppose that $\calO_{G}(v, i) = e$.
  Then, we invoke \textsf{Component}$(e)$.
  If \textsf{Component}$(e)$ returns \textbf{Large}, we return $e$.
  Otherwise, we take any base $B$ of the returned set of \textsf{Component}$(e)$ by an existing algorithm.
  We return $e$ if $e \in B$ and return $\bot$ if otherwise.
  Note that for another edge $f \in S$, we use the same base $B$.

  To analyze the query complexity, note that, during \textsf{Component}$(e)$, we perform queries $\calO_G(v,i)$ only for vertices $v$ in $G_{3t}(e)$. 
  So, to perform \textsf{Component}$(e)$, we need $d^{3t} = d^{3\ell d/\epsilon}$ queries to $\calO_G$.
  In total, we need $d^{3\ell d/\epsilon}(k+d)^{O(1/\epsilon'^2)}(\frac{1}{\epsilon'})^{O(1/\epsilon')}=(k+d)^{O(1/\epsilon'^2)}(\frac{1}{\epsilon'})^{O(1/\epsilon')}$, where $\epsilon' = \frac{\epsilon}{k+d\ell}$.
\end{proof}
Theorem~\ref{thr:test-k-l-fullness} directly follows from Theorem~\ref{thr:approximation-to-mkl}.

\section{Testing $(k, \ell)$-edge-connected-orientability}\label{sec:orientability}
In this section, we present a tester for the $(k,\ell)$-edge-connected-orientability of a graph $G=(V,E)$.  

A multiset ${\cal F}=\{V_1,\dots, V_s\}$ of subsets of $V$ is said to be {\em regular} if each element of $V$ belongs to the same number of subsets in ${\cal F}$.
For a regular multiset $\calF=\{V_1,\dots, V_s\}$ of subsets of $V$,
let $d_G(\calF)=\sum_{i=1}^s \frac{d_G(V_i)}{2}$.
If $\calF$ is a partition of $V$, $d_G(\calF)$ amounts to the number of edges connecting distinct subsets of $\calF$.

In~\cite{F80}, Frank proved a characterization of the orientability of graphs satisfying a so-called supermodular covering condition.
This theorem includes the following characterization of the $(k,\ell)$-ec-orientability as a special case (see e.g.,~\cite{FK02,FK03} for more detail).
\begin{theorem}[Frank~\cite{F80}]
\label{thm:orientability}
Let $G=(V,E)$ be a graph.
Then, $G$ admits a $(k,\ell)$-edge-connected-orientation if and only if
$d_G(\calF)\geq k(|\calF|-1)+\ell$ for any partition $\cal F$ of $V$ into non-empty subsets with $|\calF|\geq 2$.
\end{theorem}
This theorem motivates us to look at the following deficiency function:
\begin{eqnarray}
\label{eq:eta1}
  \eta_{k,\ell}(G) = \max\{0,\max\{k(|\calF| - 1) + \ell - d_G(\calF)\mid \text{ a partition $\calF$ of $V$ with $|\calF|\geq 2$}\}\},
\end{eqnarray}
Then, $G$ admits a $(k,\ell)$-ec-orientation if and only if $\eta_{k,\ell}(G)=0$.

Notice that, if $\ell=0$, we have $k(|{\calF}|-1)-d_G({\calF})=0$ for $\calF=\{V\}$. 
We thus redefine $\eta_{k,0}(G)$, for convenience, by 
\begin{equation}
\label{eq:eta2}
\eta_{k,0}(G)=\max\{k(|\calF|-1)-d_G(\calF) \mid \text{ a partition $\calF$ of $V$}\}.
\end{equation}
($\eta_{k,\ell}(G)$ remains \eqref{eq:eta1} if $\ell>0$.)
Tutte~\cite{tutte61} and Nash-Williams~\cite{nash61} proved a special relation between $\eta_{k,0}(G)$ and the arbolicity of $G$.
Specifically, Tutte-and-Nash-Williams tree packing theorem can be described in terms of ${\cal M}_{k,k}(G)$ as follows.
\begin{theorem}[Tutte~\cite{tutte61}, Nash-Williams~\cite{nash61}]
\label{thm:tutte}
Let $G=(V,E)$ be a graph and $k \geq 1$ be an integer.
Then, 
\begin{equation}
\label{eq:tutte}
\rho_{k,k}(E)=k(n-1)-\eta_{k,0}(G).
\end{equation}
\end{theorem}
Notice that 
  \begin{align*}
\eta_{k,\ell}(G)&=\max\{k(|\calF| - 1) + \ell - d_G(\calF)\mid \text{ a partition ${\cal F}$ of $V$ with $|\calF|\geq 2$}\}  \\
&= \max\{k(|\calF| - 1) - d_G(\calF) \mid \text{ a partition ${\cal F}$ of $V$ with $|\calF|\geq 2$}\} + \ell \\
&\leq \eta_{k,0}(G)+\ell,
  \end{align*}
where the equality holds if $\eta_{k,0}(G)>0$. 
Hence, we also have $\eta_{k,0}(G)\leq \eta_{k,\ell}(G)$.
Namely,
\begin{equation}
\label{eq:eta_relation}
\eta_{k,0}(G)\leq \eta_{k,\ell}(G)\leq \eta_{k,0}(G)+\ell.
\end{equation}
Since $\eta_{k,0}(G)$ can be computed from $\rho_{k,k}(G)$ by \eqref{eq:tutte},
the approximation algorithm for $\rho_{k,k}(G)$ proposed in Theorem~\ref{thr:approximation-to-mkl} can be modified to compute $\eta_{k,\ell}(G)$.
\begin{corollary}\label{crl:approximation-to-eta}
  Let $G$ be a graph with $n$ vertices, and $k \geq 1,\ell \geq 0$ be integers with $2k-\ell \geq 1$.
  In the bounded-degree model with a degree bound $d$,
  there exists a $(1,\ell+\epsilon n)$-approximation algorithm for $\eta_{k,\ell}(G)$ with query complexity $(k+d)^{O(1/\epsilon'^2)}(\frac{1}{\epsilon'})^{O(1/\epsilon')}$ where $\epsilon' = \frac{\epsilon}{dk}$.
\end{corollary}

For testing $(k,\ell)$-ec-orientability, we need a certificate for deciding whether $G$ is $\epsilon$-far from $(k,\ell)$-ec-orientable.
This part relies on a structural property of the connectivity argumentation problem proved by Frank and Kir{\'a}ly~\cite{FK03}. 
A family $\{X_1,\dots, X_s\}$ of subsets of $X\subseteq V$ is called a {\em co-partition} of $X$ if $\{V\setminus X_1, \dots, V\setminus X_s\}$ forms a partition of $V\setminus X$.
Also, for two multisets $\calF_1$ and $\calF_2$, $\calF_1+\calF_2$ denotes their union as a multiset.
\begin{theorem}[Frank and Kir{\'a}ly~\cite{FK03}]\label{thm:fk03-augment}
  Let $G=(V,E)$ be a graph.
  $G$ can be made $(k,\ell)$-ec-orientable by adding $\gamma$ edges iff the following two conditions hold:
  \begin{description}
  \item[(A)] \( \gamma \geq k(|\calF|-1)+\ell - d_G(\calF)  \) for every partition $\calF$ of $V$ with $|{\cal F}|\geq 2$.
  \item[(B)] \( 2\gamma \geq |\calF_1|k + |\calF_2|\ell - d_G(\calF)\) for every multiset $\calF = \calF_1+ \calF_2$ satisfying the following three conditions:
\begin{align}
&\text{$\calF_1$ is a partition of some $X \subset V$,} \nonumber \\
&\text{$\calF_2$ is a co-partition of $V\setminus X$,} \label{eq:partition} \\
&\text{every member of $\calF_2$ is the complement of the union of some members of $\calF_1$.} \nonumber
\end{align}
\end{description}
\end{theorem}

By Corollary~\ref{crl:approximation-to-eta}, the condition (A) is efficiently checkable.
The non-trivial part is an algorithm for checking the second condition.
Let
\begin{equation}
\xi_{k,\ell}(G)=\max_{\calF=\calF_1+\calF_2}\{|\calF_1|k + |\calF_2|\ell - d_G(\calF)\} 
\end{equation}
where the maximum is taken over all multisets $\calF = \calF_1+\calF_2$ satisfying \eqref{eq:partition}.
Our goal is to approximate $\xi_{k,\ell}(G)$ efficiently.
To simplify $\xi_{k,\ell}$, we need some terminology.
For two partitions $\calP_1$ and $\calP_2$ of $X\subseteq V$, $\calP_1$ is said to be a {\em refinement} of $\calP_1$ if each member of $\calP_2$ is the union of some members of $\calP_1$.
A regular multiset $\calP$ of subsets of $X$ is called a {\em double-partition} of $X$ if $\calP$ is written as $\calP=\calP_1+\calP_2$ for some partitions $\calP_1$ and $\calP_2$ of $X$ such that 
$\calP_1$ is a refinement of $\calP_2$. 

We should note the following relation between a double-partition and a multiset satisfying \eqref{eq:partition}.
Let $\calF=\calF_1+\calF_2$ be a family of subsets satisfying \eqref{eq:partition} with a partition $\calF_1$ of $X$ and a co-partition $\calF_2$ of $V\setminus X$.
Let $\calP_1=\calF_1$ and $\calP_2=\{V\setminus X'\mid X'\in \calF_2\}$.
Then, $\calP_1$ and $\calP_2$ are partitions of $X$ and $\calP_1$ is a refinement of $\calP_2$.
Also, carefully counting the number of edges contributed to $d_G({\cal F})$, we have $d_G(\calF)=d_G(\calP_1+\calP_2)$. 
Thus, using 
\begin{align*}
\eta_{k,0}(G)&=\max\{k(|\calP|-1)-d_G({\calP})\mid \text{ a partition $\calP$ of $V$} \} \\
&= \max\left\{k(|\calP|-1)-\sum_{i=1}^s \frac{d_G(X_i)}{2}\mid \text{ a partition $\calP=\{X_1,\dots, X_s\}$ of $V$} \right\} \\
&= \max\left\{\sum_{i=1}^s \left(k-\frac{d_G(X_i)}{2}\right)-k\mid \text{ a partition $\calP=\{X_1,\dots, X_s\}$ of $V$} \right\}
\end{align*}
we obtain
\begin{align*}
\xi_{k,\ell}(G)&=\max\{|\calF_1|k + |\calF_2|\ell - d_G(\calF)\mid \text{ a family }\calF=\calF_1+\calF_2 \text{ satisfying \eqref{eq:partition}}\} \\
&=\max\{|\calP_1|k + |\calP_2|\ell - d_G(\calP)\mid \text{ a double-partition } \calP=\calP_1+\calP_2 \text{ of some } X\subset V \}  \\
\begin{split}&=\max \bigg\{\sum_{i=1}^s \bigg(\ell-d_G(X_i)+\max\bigg\{ \sum_{j=1}^{s_i} (k-d_{G[X_i]}(X_{i,j})/2) \mid \text{ a partition $\{X_{i,1},\dots, X_{i,s_i}\}$ of $X_i$} \bigg\}\bigg) \\ 
& \hspace{25em} \mid \text{ a sub-partition $\{X_1,\dots,X_s\}$ of $V$} \bigg\} \end{split} \\
&=\max\bigg\{\sum_{i=1}^s (k+\ell-d_G(X_i)+\eta_{k,0}(G[X_i])) \mid \text{ a sub-partition $\{X_1,\dots, X_s\}$ of $V$} \bigg\}.
\end{align*}
Let $g_{k,\ell}(X)=k+\ell-d(X)+\eta_{k,0}(G[X])$ for $X\subseteq V$.
Then, we have 
\begin{equation}
\label{eq:xi}
 \xi_{k,\ell}(G)=\max\bigg\{\sum_{i=1}^s g_{k,\ell}(X_i)\mid \text{ a sub-partition ${\cal P}=\{X_1,\dots, X_s\}$ of $V$} \bigg\}. 
\end{equation}

We say that $X\subseteq V$ is {\em deficient} if $g_{k,\ell}(X)>0$. 
By Theorem~\ref{thm:fk03-augment} and \eqref{eq:xi},  $g_{k,\ell}(X)\leq 0$ holds for every $X$ with $\emptyset\neq X\subsetneq V$ if $G$ is $(k,\ell)$-ec-orientable.
The following theorem is a key result for developing a constant-time tester.
\begin{theorem}
\label{thm:key2}
For a given $\epsilon$, let $c=\frac{\epsilon^2d^2}{16k\ell}$ and $t=\frac{4\ell}{\epsilon d}$. 
Suppose that $\xi_{k,\ell}(G)\geq \epsilon dn$. Then, at least one of the followings holds:
\begin{description}
\item[(i)] There are at least $cn$ disjoint small deficient sets, where a set is called small if the cardinality is less than $t$;
\item[(ii)] $\eta_{k,0}(G)\geq \frac{1}{4}\epsilon dn$. Namely, $G$ is $\frac{\epsilon}{2}$-far from $(k,k)$-fullness. 
\end{description}
\end{theorem}
\begin{proof}
Let ${\cal P}=\{X_1,\dots, X_s\}$ be a sub-partition of $V$ that maximizes the right hand side of \eqref{eq:xi}.
Since the maximum is taken over all sub-partitions of $V$, we may assume $g_{k,\ell}(X_i)>0$ for all $1\leq i\leq t$. 
Let us divide ${\cal P}$ into two subsets ${\cal P}_{\rm small}$ and ${\cal P}_{\rm large}$ depending on whether it is small or not.
Notice that for each $X\in {\cal P}_{\rm small}$ we have 
$g_{k,\ell}(X)=k+\ell-d_G(X)+\eta_{k,0}(G[X])\leq kt$ since $\eta_{k,0}(G[X])\leq k(|X|-1)\leq kt-2k$.

Suppose that (i) does not happen. Then, by $\xi_{k,\ell}(G)\geq \epsilon dn$ and  $\sum_{X\in {\cal P}_{\rm small}} g_{k,\ell}(X)\leq ktcn$, we have
\begin{equation}
\label{eq:large_deficient1}
\sum_{X\in {\cal P}_{\rm large}}g_{k,\ell}(X)\geq (\epsilon d-ktc)n.
\end{equation}
We now prove the following relation between $\eta_{k,0}$ and $g_{k,\ell}$, which gives us a lower bound on $\eta_{k,0}(G)$.
\begin{equation}
\label{eq:large_deficient2}
\eta_{k,0}(G)\geq \frac{1}{2}\sum_{X\in {\cal P}_{\rm large}} (g_{k,\ell}(X)-\ell).
\end{equation}
Recall that $\eta_{k,0}(G)$ is the number of edges we need to add to make $G$ $(k,k)$-full.
Hence, we can take  a new graph $H=(V,E_H)$ on $V$  such that $|E_H|=\eta_{k,0}(G)$ and $G'=(V,E\cup E_H)$  is $(k,k)$-full.
We need the following formulae.
\begin{claim}
\label{claim:5_1}
For any $X\subseteq V$,
\begin{description}
\item[(a)] $\eta_{k,0}(G'[X])+k\leq d_{G'}(X)$, and
\item[(b)] $\eta_{k,0}(G[X])\leq \eta_{k,0}(G'[X])+i_{H}(X)$, where $i_H(X)=|\{uv\in E_H\mid u,v\in X\}|$. 
\end{description}
\end{claim}
\begin{proof}
Let ${\cal P}_X'$ be a partition of $X$ such that 
$\eta_{k,0}(G'[X])=k(|{\cal P}_X'|-1)-d_{G'[X]}({\cal P}_X')$.
Let ${\cal F}={\cal P}_X'\cup\{V\setminus X\}$. 
Then $\calF$ is a partition of $V$.
Since $G'$ is $(k,k)$-full, we have
\begin{align*}
0=\eta_{k,0}(G')&\geq k(|\calF|-1)-d_{G'}(\calF) \\
&=k(|\calP_X'|-1)+k-d_{G'}(X)-d_{G'[X]}(\calP_X') \\
&=\eta_{k,0}(G'[X])+k-d_{G'}(X),
\end{align*}
implying (a).

On the other hand, let ${\cal P}_X$ be a partition of $X$ such that 
$\eta_{k,0}(G[X])=k(|{\cal P}_X|-1)-d_{G[X]}({\cal P}_X)$.
Since $d_{G[X]}({\cal P}_X)+i_H(X)\geq d_{G'[X]}({\cal P}_X)$, we have
\begin{align*}
\eta_{k,0}(G[X])&=k(|\calP_X|-1)-d_{G[X]}(\calP_X) \\
&\leq k(|\calP_X|-1)-d_{G'[X]}(\calP_X)+i_H(X) \\
&\leq \eta_{k,0}(G'[X])+i_H(X).
\end{align*}
\end{proof}
In total, we have
\begin{align*}
\sum_{X\in \calP_{\rm large}}(g_{k,\ell}(X)-\ell)&=\sum_{X\in {\cal P}_{\rm large}}(\eta_{k,0}(G[X])+k-d_G(X)) \\
&\leq \sum_{X\in {\cal P}_{\rm large}}(\eta_{k,0}(G'[X])+i_H(X)+k-d_G(X)) \\
&\leq \sum_{X\in {\cal P}_{\rm large}}(d_{G'}(X)+i_H(X)-d_G(X)) \\
&=  \sum_{X\in {\cal P}_{\rm large}}(d_{H}(X)+i_H(X)) \\
&\leq 2|E_H|=2\eta_{k,0}(G).
\end{align*}
Thus, we obtain \eqref{eq:large_deficient2}.
Moreover, since there are at most $\frac{n}{t}$ large sets among ${\cal P}$, \eqref{eq:large_deficient2} implies
\begin{equation}
\label{eq:large_deficient3}
\eta_{k,0}(G)\geq \sum_{X\in {\cal P}_{\rm large}} \frac{g_{k,\ell}(X)}{2} -\frac{\ell n}{2t}.
\end{equation}
Combining \eqref{eq:large_deficient1} and \eqref{eq:large_deficient3}, we finally have
$\eta_{k,0}(G)\geq \frac{1}{2}(\epsilon d-ktc-\frac{\ell}{t})n=\frac{\epsilon dn}{4}$.
This completes the proof.
\end{proof}

A testing algorithm for the $(k,\ell)$-ec-orientability of a graph $G=(V,E)$ is given in Algorithm~3.
In Line~7, $V_t(v)$ denotes the set of vertices whose distances to $v\in V$ are at most $t$.
\begin{algorithm}
  \caption{Testing the $(k,\ell)$-ec-orientability of a bounded-degree graph $G$}
  \label{alg:k-l-orientability}
  \begin{algorithmic}[1]
    \STATE Take any $\epsilon''$ such that $\epsilon''<\epsilon$.
    \STATE Run a $(1,\frac{\epsilon'' dn}{4})$-approximation algorithm for $\eta_{k,0}(G)$. 
    \IF{the obtained value $x^*$ satisfies $x^*>0$}
    \STATE \textbf{reject} $G$.
    \ENDIF
    \STATE Choose a set $S$ of $\frac{8k\ell}{\epsilon^2d^2}$ vertices uniformly at random from $G$. 
    \FOR{$v\in S$}
    \STATE compute $X_v={\rm argmax}\{g_{k,\ell}(X):X\subseteq V_{t}(v), X\neq \emptyset\}$ with $t=\frac{4\ell}{\epsilon d}$.
    \IF{$g_{k,\ell}(X_v)>0$}
    \STATE \textbf{reject} $G$.
    \ENDIF
\ENDFOR
    \STATE \textbf{accept} $G$.
  \end{algorithmic}
\end{algorithm}

\begin{proof}[Proof of Theorem~\ref{thr:test-k-l-orientability}]
We prove that Algorithm~\ref{alg:k-l-orientability} can be implemented with query complexity $(k+d)^{O(1/\epsilon'^2)}(\frac{1}{\epsilon'})^{O(1/\epsilon')}$ where $\epsilon'=\max(\frac{\epsilon }{dk},\frac{d \epsilon }{\ell})$, and correctly tests the $(k,\ell)$-ec-orientability of $G$.

For Line~2, we use an approximation algorithm mentioned in Corollary~\ref{crl:approximation-to-eta} that runs in $(k+d)^{O(1/\epsilon'^2)}(\frac{1}{\epsilon'})^{O(1/\epsilon')}$ queries.
For Line~7, we use an algorithm for minimizing $g_{k,\ell}$ on $V_t(v)$.
Let $\hat{g}_{k,\ell}:2^{V_t(v)}\rightarrow \mathbb{Z}$ be the function defined by, for each $X\subseteq V_t(v)$,
\begin{equation*}
\hat{g}_{k,\ell}(X)=\begin{cases}
-\infty & \text{if } X=\emptyset \\
g_{k,\ell}(X) & \text{otherwise}.
\end{cases}
\end{equation*}
Since $g_{k,\ell}$ is a supermodular function (see Appendix~\ref{sec:supermodularity}),
it is easy to observe that $\hat{g}_{k,\ell}$ is an intersecting supermodular function 
(i.e., $f(X)+f(Y)\leq f(X\cap Y)+f(X\cup Y)$ holds for any $X,Y\subseteq V_t(v)$ with $X\cap Y\neq \emptyset$).
Thus, to perform Line~7, we call a polynomial time algorithm for minimizing an intersecting submodular function (see e.g.,~\cite{Schrijver, fujishige}).
By $|V_t(v)|\leq d^t$, the query complexity taken in the for-loop is upper bounded by $O(\frac{d^2}{\epsilon'^2})\cdot d^{O(1/\epsilon')}$.  
Thus, Algorithm~\ref{alg:k-l-orientability} can be implemented with query complexity $(k+d)^{O(1/\epsilon'^2)}(\frac{1}{\epsilon'})^{O(1/\epsilon')}$.

To see the correctness, assume first that $G$ is $(k,\ell)$-ec-orientable.
Then, we have $0=\eta_{k,\ell}(G)\geq \eta_{k,0}(G)$ by (\ref{eq:eta_relation}).
Since $x^*\leq \eta_{k,0}(G)$ with probability at least $\frac{2}{3}$,  $x^*\leq 0$ holds. 
Namely, the algorithm does not reject $G$ at Line~4 with probability $\frac{2}{3}$.
Also, since $g_{k,\ell}(G)\leq 0$ for every $X$ with $\emptyset\neq X\subsetneq V$ by Theorem~\ref{thm:fk03-augment},
the algorithm never rejects $G$ at Line~9.
Thus, Algorithm~3 accepts $G$ with probability at least $\frac{2}{3}$.

Conversely, suppose that $G$ is $\epsilon$-far.
Then, by Theorem~\ref{thm:fk03-augment} and (\ref{eq:eta_relation}),
$\eta_{k,0}(G)+\ell\geq \eta_{k,\ell}(G)\geq \frac{\epsilon dn}{2}$ or $\xi_{k,\ell}(G)\geq \epsilon dn$ holds.
Since $\ell <\frac{\epsilon dn}{4}$ (otherwise $n$ becomes constant and we can apply any existing polynomial time algorithm),
$\eta_{k,0}(G)\geq \frac{\epsilon dn}{4}$ or $\xi_{k,\ell}(G)\geq \epsilon dn$ holds.
In total, combining this with Theorem~\ref{thm:key2},
$\eta_{k,0}(G)\geq \frac{\epsilon dn}{4}$ holds or $G$ has  at least $cn$ disjoint deficient sets of size at most $t$, where $c=\frac{\epsilon^2d^2}{16k\ell}$.
If $\eta_{k,0}(G)\geq \frac{\epsilon dn}{4}$, then $x^*\geq \eta_{k,0}(G)-\frac{\epsilon''dn}{4}>0$ holds and $G$ is rejected in Line~4;
Otherwise, the probability that we choose some vertex in a deficient set of size at most $t$ in Line~5 is at least
\begin{equation}
1-\left(\frac{n-cn}{n}\right)^{\frac{2}{c}}\geq 1-\frac{1}{e^2}\geq \frac{2}{3},
\end{equation}
and  Algorithm~3 rejects $G$ with probability at least $\frac{2}{3}$ in Line~9.
\end{proof}

\section{Linear Lower Bounds for One-Sided Error Testers}\label{sec:lower-bound}
In this section, we prove Theorems~\ref{thr:lower-k-l-fullness} and~\ref{thr:lower-k-l-ec-orientability}.
As for $(k,\ell)$-ec-orientability,
in the bounded-degree model,
Orenstein~\cite{Ore10} showed a liner lower bound of one-sided error tester for $(k,0)$-ec-orientability where $k \geq 2$.
We can easily modify his proof to achieve Theorem~\ref{thr:lower-k-l-ec-orientability} by using Theorem~\ref{thm:orientability}.
Thus, we omit the detail in this paper.
He also showed that there is a one-sided error tester for $(1,0)$-ec-orientability.
We cannot extend the lower bound to the case $\ell \geq k$ since, in such a case, 
$(k,\ell)$-ec-orientability coincides with the $(k+\ell)$-edge-connectivity,
and we have one-sided error testers for it~\cite{GR02}.

In what follows, we consider lower bounds for testing $(k,\ell)$-fullness.
Let $k \geq 2,\ell \geq 0$ be integers with $2k - \ell \geq 1$.
We mention that, when $k = 1$,
it is easy to make one-sided error testers (see Appendix~\ref{sec:trivial-testers}).
Note that a one-sided error tester cannot reject a graph until it has found an evidence that the graph is not $(k,\ell)$-full,
i.e., an $\epsilon$-far graph cannot be $(k,\ell)$-full no matter how we add edges in the unseen part of the graph.
With this observation, 
Orenstein~\cite{Ore10} constructed a graph which is $\epsilon$-far from $(k,k)$-full while if one has seen only $\beta n$ vertices for some constant $\beta$, one can add edges so that the resulting graph is $(k,k)$-full.
His construction relies on the Tutte-Nash-Williams tree-packing theorem (Theorem~\ref{thm:tutte}), which is a special property of $(k,k)$-fullness.
We complete the work by showing the existence of such a graph for general $(k,\ell)$-fullness.

First, we define a \textit{$(\beta,\gamma)$-expander} as a graph $G=(V,E)$ such that for any $S \subseteq V, |S| \leq \beta |V|$,
we have that $|\Gamma(S)| \geq \gamma |S|$.
The following lemma states that such graphs indeed exist.
\begin{lemma}[See e.g.,~\cite{HLW06}]\label{lmm:expander}
  Let $d \geq 3$ be an integer.
  Then, there exists a $d$-regular $(\beta,d-2)$-expander graph for some universal constant $\beta > 0$.
\end{lemma}

\begin{lemma}\label{lmm:far}
  Let $G$ be the $(2k-1)$-regular expander graph of $n$ vertices given in Lemma~\ref{lmm:expander}.
  When $n$ is sufficiently large,
  $G$ is $\epsilon$-far from $(k,\ell)$-fullness for $\epsilon = O(\frac{1}{k})$. 
\end{lemma}
\begin{proof}
  Note that any $(k,\ell)$-full graph must have at least $kn - \ell$ edges.
  However, $G$ has $\frac{(2k-1)n}{2}$ edges.
  Thus, to make $G$ $(k,\ell)$-full,
  we need to add at least $kn - \ell - \frac{(2k-1)n}{2} = \frac{n}{2} - \ell$ edges.
  Thus, the lemma holds.
\end{proof}
The following is a well-known graph operation that preserves $(k,\ell)$-fullness. 
\begin{lemma}[See e.g.,~\cite{Fekete:2004}]\label{lmm:combine}
  Let $G = (V,E)$ be a $(k,\ell)$-full graph.
  We introduce a new vertex $v$ and connect $v$ and distinct $k$ vertices of $V$ by new edges.
  Then, the resulting graph is also $(k,\ell)$-full.
\end{lemma}
We also need the following graph operation for the case $k=2$.
\begin{lemma}
\label{lmm:combine2}
Let $G=(V,E)$ be a $(2,\ell)$-full graph with $|V|\geq 2$.
We introduce a cycle graph $G'=(U,C)$ consisting of new vertices $U=\{u_1,u_2,\dots, u_s\}$ and then connect each new vertex $u_i$ to a vertex $v_i\in V$ so that
$v_i\neq v_j$ for some $1\leq i<j\leq s$.
Then, the resulting graph is also $(2,\ell)$-full.
\end{lemma}
\begin{proof}
It is sufficient to consider the case when $G$ is $(2,\ell)$-tight.
Let $E_{U,V}=\{u_iv_i\mid 1\leq i\leq s\}$.
Note that the total number of edges amounts to $|E|+|C|+|E_{U,V}|=2|V|-\ell+|U|+|U|=2|V\cup U|-\ell$.

Suppose that the resulting graph is not $(2,\ell)$-tight.
Then, there is an edge subset $F$ that violates the counting condition, i.e., $|F|>f_{2,\ell}(F)$.
We split $F$ into three parts; $F_{U}=F\cap C$, $F_{U,V}=F\cap E_{U,V}$ and $F_V=F\cap E$.
Since $G'$ is a cycle, we have $|F_U|\leq |V_{G'}(F)|$.
Also, each vertex $u_i\in U$ is incident to only one vertex in $V$, $|F_{U,V}|\leq |V_{G'}(F)|$.
Thus, if $F_V\neq \emptyset$, we have $|F| = |F_U| + |F_{U,V}| + |F_V| \leq |V_{G'}(F)| + |V_{G'}(F)| + 2|V_G(F)| - \ell = 2|V_{G'}(F)\cup V_G(F)| - \ell=f_{2,\ell}(F)$.
Therefore, $F_V=\emptyset$ must hold, but a simple counting argument shows that any subset of $C\cup E_{U,V}$ cannot violate the counting condition, which is a contradiction.
\end{proof}

\begin{proof}[Proof of Theorem~\ref{thr:lower-k-l-fullness}]
  Let $G = (V,E_G)$ be the $(2k-1)$-regular $(\beta,2k-3)$-expander graph of $n$ vertices given in Lemma~\ref{lmm:expander}.
  From Lemma~\ref{lmm:far}, $G$ is $O(\frac{1}{k})$-far from $(k,\ell)$-fullness.
  Suppose that an algorithm $\calA$ has queried $\beta n$ times.
  and let $V_{\calA} \subseteq V$ be the set of vertices involved with those queries.
  That is, for every $v \in V_{\calA}$, 
  there was a query of the form $\calO_G(v,i)$ for some $i \in [d]$,
  or $\calO_G$ returns an edge incident to $v$.
  Clearly, $|V_{\calA}| \leq \beta n$ holds.

  Let $S = V \setminus V_{\calA}$.
  We take any $(k,\ell)$-full graph $H=(S,E_H)$ on $S$ using new edges.
  Then, we consider the graph $G' = (V, E_G \cup E_H)$.
  We show that $G'$ is $(k,\ell)$-full.
  This means that any algorithm cannot reject $G$ just by seeing $\beta n$ edges.
 
  We know that $G'[S]$ is $(k,\ell)$-full since $H$ is $(k,\ell)$-full.
  To show that entire $G'$ is $(k,\ell)$-full,
  we iteratively enlarge $S$ keeping that $G'[S]$ is $(k,\ell)$-full. Let $\overline{S}=V\setminus S$.
  We have the following two cases.
  \begin{itemize}
  \item If there exists a vertex $v \in \overline{S}$ such that $|\Gamma(v) \cap S| \geq k$,  
    then $G'[S+v]$ is also $(k,\ell)$-full by Lemma~\ref{lmm:combine}.
    Thus, we replace $S$ by $S + v$.
  \item If every vertex $v \in \overline{S}$ satisfies $|\Gamma(v) \cap S| < k$, then we have $d_{G'}(v,S) < k$ for any $v \in \overline{S}$.
    Since $G'$ contains a $(\beta,2k-3)$-expander,
    we have that $d_{G'}(\overline{S}) \geq (2k-3)|\overline{S}|$.
    However, the assumption implies that $d_{G'}(\overline{S}) = \sum_{v \in \overline{S}}d_{G'}(v,S) < k|\overline{S}|$.
    Combining those inequalities, we have $k = 2$.
    Furthermore, by $d_{G'}(v,S)<k$, we have $d_{G'}(v,S) = 1$ for every $v \in \overline{S}$.
    Note that the degree of any vertex $v \in \overline{S}$ is $2k-1 = 3$.
    This means that $G'[\overline{S}]$ consists of disjoint cycles and we have an edge from each vertex $v \in \overline{S}$ to $S$.
    Let us take such a cycle and let $U$ be the vertex set of this cycle. Then, by $|\Gamma(U)| \geq |U|$, we can apply Lemma~\ref{lmm:combine} to claim that $G'[S\cup U]$ is $(2,\ell)$-full.
    Thus, we replace $S$ by $S\cup U$.
  \end{itemize}
  For any of those two cases, we can enlarge $S$ until $S$ becomes $V$.
  Thus, the theorem holds.
\end{proof}

\section{Concluding Remarks}
The concept of $(k,\ell)$-sparsity can be generalized as follows.
For a hypergraph $H=(V,{\cal E})$, let $\mathbf{k}:V\rightarrow \mathbb{Z}_+$ and $\ell\in \mathbb{Z}_+$.
We define a function $f_{\mathbf{k},\ell}:2^{\cal E}\rightarrow \mathbb{Z}_+$ by $f_{\mathbf{k},\ell}({\cal E}')=\sum_{v\in V({\cal E}')}\mathbf{k}(v)-\ell$ for ${\cal E}'\subseteq {\cal E}$,
where $V({\cal E}')=\bigcup_{X\in {\cal E}'}X$.
It is easy to see that $f_{\mathbf{k},\ell}$ is non-decreasing and submodular, and thus $f_{\mathbf{k},\ell}$ induces a matroid, ${\cal M}_{\mathbf{k},\ell}(H)$, on ${\cal E}$.
%
It is easy to generalize our approximation algorithm to that for the rank of ${\cal M}_{\mathbf{k},\ell}(H)$ by just modifying the auxiliary graph $G_k$ defined in Section~\ref{sec:matching}.

Our tester for the $(k,\ell)$-fullness of a graph $G$ approximates the rank of $\calM_{k,\ell}(G)$.
It might be interesting to know for which matroid we can approximate the rank of it in constant time.
In particular, can we approximate the rank of a matrix with entries in $\mathbb{F}_2$?

We note that the $(k,k)$-fullness of a graph can be decided by checking the rank of the union of $k$ graphic matroids.
This problem is usually solved via a matroid intersection problem.
This leaves us questions: for which matroids can we approximate the rank of their union, 
and for which matroids $\calM_1,\calM_2$ can we approximate the size of the largest common independent set in $\calM_1$ and $\calM_2$ in constant queries?

\bibliography{testing_rigidity}{}

\appendix

\section{One-Sided Error Testers for $(1,0)$-Fullness and $(1,1)$-Fullness}\label{sec:trivial-testers}
In this section, we give one-sided error testers for $(1,0)$-fullness and $(1,1)$-fullness.

\begin{lemma}\label{lmm:far-from-(1,0)-full}
  Let $G$ be a graph $\epsilon$-far from $(1,0)$-fullness.
  Then, there are at least $\frac{\epsilon dn}{4}$ connected components of size at most $\frac{4}{\epsilon d}$ containing no cycle.
\end{lemma}
\begin{proof}
  Note that a graph is $(1,0)$-full iff each connected component in the graph contains a cycle.
  Thus, there are at least $\frac{\epsilon dn}{2}$ connected components containing no cycle in $G$.
  Then, it is easy to observe that the lemma holds.
\end{proof}

\begin{lemma}\label{lmm:far-from-(1,1)-full}
  Let $G$ be a graph $\epsilon$-far from $(1,1)$-fullness.
  Then, there are at least $\frac{\epsilon dn}{4}$ connected components of size at most $\frac{4}{\epsilon d}$.
\end{lemma}
\begin{proof}
  Note that a graph is $(1,1)$-full iff the graph is connected.
  Thus, there are at least $\frac{\epsilon dn}{2}$ connected components in $G$.
  Then, it is easy to observe that the lemma holds.
\end{proof}

\begin{theorem}
  In the bounded-degree model with a degree bound $d$,
  There are one-sided error testers for $(1,0)$-fullness and $(1,1)$-fullness with query complexity $O(\frac{1}{\epsilon^2 d})$.
\end{theorem}
\begin{proof}
  We describe the algorithm for $(1,1)$-fullness.
  Let $S$ be a set of $\frac{8}{\epsilon d} $ vertices chosen uniformly at random from an input graph.
  For each chosen vertex,
  we perform BFS from the vertex until we reach $\frac{4}{\epsilon d} + 1$ vertices.
  If the BFS cannot reach $\frac{4}{\epsilon d} + 1$ vertices for some vertex in $S$, we reject the graph.
  Otherwise, we accept the graph.
  Clearly, the query complexity of the algorithm is at most $\frac{8}{\epsilon d} \cdot (\frac{4}{\epsilon d}+1) \cdot d = O(\frac{1}{\epsilon^2 d})$.

  Since $(1,1)$-full graph is connected, we can reach $n$ vertices from any vertex.
  Thus, the algorithm always accepts $(1,1)$-full graph.
  Suppose that a graph is $\epsilon$-far from $(1,1)$-fullness.
  From Lemma~\ref{lmm:far-from-(1,1)-full},
  the probability that we choose some vertex in a connected component of size at most $\frac{4}{\epsilon d}$ is at least
  \begin{eqnarray*}
    1 - \left(1 - \frac{\epsilon dn}{4n}\right)^{\frac{8}{\epsilon d}} \geq 1 - \frac{1}{e^{2}} \geq \frac{2}{3}.
  \end{eqnarray*}
  For such vertex, the BFS cannot reach $\frac{4}{\epsilon d}$ vertices.
  Thus, the algorithm rejects a graph $\epsilon$-far from $(1,1)$-fullness with probability at least $\frac{2}{3}$.

  We can construct a tester for $(1,0)$-fullness in a similar way using Lemma~\ref{lmm:far-from-(1,0)-full}.
\end{proof}

\section{Supermodularity of $g_{k,\ell}$}\label{sec:supermodularity}
For each $X\subseteq V$, we have
\begin{align*}
g_{k,\ell}(X)&=k+\ell-d_G(X)+\eta_{k,0}(G[X_i]) \\
&=k+\ell-\frac{d_G(X)}{2}+\max\left\{\sum_{j=1}^s\left(k-\frac{d_G(X_j)}{2}\right)\mid \text{ a partition } \{X_1,\dots, X_s\} \text{ of } X \right\} \\
&=k+\ell-\frac{d_G(X)}{2}+\max\left\{\sum_{j=1}^s h(X_j)\mid \text{ a partition } \{X_1,\dots, X_s\} \text{ of } X \right\},
\end{align*}
where $h:2^{V}\rightarrow \mathbb{Z}$ is defined as $h(X):=k-\frac{d_G(X)}{2}$ for $X\subseteq V$.
Note that $h$ is a supermodular function and 
$\hat{h}(X):=\max\left\{\sum_{j=1}^s h(X_j)\mid \text{ a partition } \{X_1,\dots, X_s\} \text{ of } X \right\}$ is the so-called Dilworth truncation of $h$,
which is known to be supermodular again (see e.g., \cite[Chapter~48]{Schrijver}).
Since $d_G$ is submodular, $g_{k,\ell}$ is a supermodular.
\end{document}